%% file: paper.tex
\documentclass[a4paper,USenglish,runningheads]{llncs}

\input{preamble}

\usepackage{verbatim}
\usepackage{tikz}
\usetikzlibrary{arrows,shapes,decorations,automata,backgrounds,petri}
\usepackage[colorlinks=true]{hyperref}

\usepackage[capitalise]{cleveref}
\Crefname{section}{Sect.}{Sects.}
\usepackage{booktabs}

\usepackage{makecell}

\begin{document}

	\title{Parameterized Verification of\\ Round-based Distributed Algorithms\\ via Extended Threshold Automata}
	\titlerunning{Parameterized Verification via Extended Threshold Automata}

	\authorrunning{T. Baumeister, P. Eichler, S. Jacobs, M. Sakr, and M. V{\"o}lp}
	
	\author{%
		Tom	Baumeister\inst{1}\href{https://orcid.org/0009-0009-8539-6246}{\includegraphics[scale=0.03]{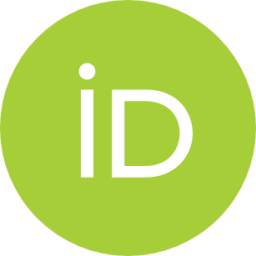}}\and 
		Paul Eichler\inst{1}\href{https://orcid.org/0009-0008-6117-318X}{\includegraphics[scale=0.03]{orcid.png}}\and
		Swen Jacobs\inst{1}\href{https://orcid.org/0000-0002-9051-4050}{\includegraphics[scale=0.03]{orcid.png}} \and 
		Mouhammad Sakr\inst{2}\href{https://orcid.org/0000-0002-5160-0327}{\includegraphics[scale=0.03]{orcid.png}} \and 
		Marcus Völp\inst{2}\href{https://orcid.org/0000-0002-8020-4446}{\includegraphics[scale=0.03]{orcid.png}} 
	}%
	\institute{
		CISPA Helmholtz Center for Information Security, Germany\\
		\and
		SnT, Luxembourg University}

\maketitle              
 \begin{abstract} 
   Threshold automata are a computational model that has proven to be versatile in modeling threshold-based distributed algorithms and
   enabling their completely automatic parameterized verification.
We present novel techniques for the verification of threshold automata, based on well-structured transition systems, that allow us to extend the expressiveness of both the computational model and the specifications that can be verified.
In particular, we extend the model to allow decrements and resets of shared variables,
possibly on cycles, and the specifications to general coverability.
While these extensions of the model in general lead to undecidability, our algorithms provide a semi-decision procedure.
We demonstrate the benefit of our extensions by showing that we can model complex round-based algorithms such as the phase king consensus algorithm and the Red Belly Blockchain protocol (published in 2019), and verify them fully automatically for the first time.

\end{abstract}

\input{intro}

\input{system-model}
\input{specs}
\input{spurious}

\input{Thresh-Autom-WSTS}
\input{zcs}

\input{reachability}

\input{experiments}
\input{related-work}
\input{conclusion}

\bibliographystyle{splncs04}
\bibliography{FTDA}
\newpage
\appendix
\input{proofs}
\input{ReachAlg}
\input{zcsCover}

\end{document}

%% file: preamble.tex
\newif\iffinal
\finaltrue

\newif\ifinline
\iffinal
\else
\inlinetrue
\fi

\newif\ifwithextensions
\withextensionsfalse
\usepackage{xcolor}
\usepackage{amsfonts}
\usepackage{amsmath,thm-restate}
\usepackage{amssymb}
\usepackage{temporal}
\usepackage{cite}
\usepackage{graphicx}
\usepackage{enumitem}
\usepackage{wrapfig}
\usepackage{environ}
\usepackage{multirow}
\usepackage{xspace}
\usepackage{color,soul}
  \definecolor{lightblue}{rgb}{.8,.95,1}

\usepackage{thmtools, thm-restate}

\usepackage{tikz}
\usepackage{pgf}
\usetikzlibrary{positioning,calc,automata,arrows,fit,shapes}

\usepackage{caption}
\usepackage{mdframed}

\usepackage{todonotes}
\usepackage{soul}
\definecolor{mygreen}{RGB}{0,150,0}

\newcommand{\mailto}[1]{\href{mailto:#1}{\nolinkurl{#1}}}

\renewcommand{\vec}[1]{\mathbf{#1}}
\newcommand{\mT}{\mathcal{T}}
\newcommand{\mTz}{\mathcal{T}^{z}}
\newcommand{\mR}{\mathcal{R}}
\newcommand{\mI}{\mathcal{I}}

\newcommand{\mD}{\mathcal{D}}

\newcommand{\mTH}{\mathcal{TH}}

\newcommand{\resets}{\tau}
\newcommand{\Shared}{\Gamma}
\newcommand{\vP}{\vec{P}}
\newcommand{\vp}{\vec{p}}
\newcommand{\from}{\textit{from}}
\renewcommand{\to}{\textit{to}}
\newcommand{\vup}{\vec{uv}}
\newcommand{\Rules}{\mR}

\newcommand{\cs}{\mathsf{CS}}
\newcommand{\csa}{\mathsf{CS(A)}}
\newcommand{\vg}{\vec{g}}
\newcommand{\vk}{\vec{k}}
\newcommand{\Confs}{\Sigma}
\newcommand{\conf}{\sigma}
\newcommand{\Trans}{\mT}

\newcommand{\absTA}{\overline{\text{TA}}}
\newcommand{\absA}{\overline{A}}
\newcommand{\absShared}{\overline{\Shared}}
\newcommand{\absRules}{\overline{\Rules}}
\newcommand{\bx}{\overline{x}}
\newcommand{\ar}{\overline{r}}

\newcommand{\acs}{\mathsf{ACS}}
\newcommand{\acsa}{\mathsf{ACS(\absA)}}
\newcommand{\avg}{\overline{\vg}}
\newcommand{\absConfs}{\overline{\Confs}}
\newcommand{\absconf}{\overline{\conf}}
\newcommand{\absTrans}{\overline{\Trans}}
\newcommand{\abcs}{(\absConfs, \absConfs_0, \absTrans)}
\newcommand{\bvarphi}{\ensuremath{\overline{\varphi}}}
\newcommand{\bphi}{\ensuremath{\overline{\varphi}}}

\newcommand{\zcs}{\mathsf{ZCS}}
\newcommand{\zcsa}{\mathsf{ZCS(\absA)}}
\newcommand{\zConfs}{\Confs^z}
\newcommand{\zconf}{{\conf^{z}}}
\newcommand{\zpi}{\mathsf{\pi^z}}
\newcommand{\vkz}{\vec{k}^{z}}

\newcommand{\vu}{\vec{u}}

\newcommand{\bpi}{\mathsf{\bar{\pi}}}

\definecolor{darkgreen}{rgb}{0,0.5,0}
\definecolor{darkblue}{rgb}{0,0,.5}
\definecolor{mygray}{gray}{.3}

\newcommand{\state}{q}

\newcommand{\cstate}{s}
\newcommand{\cstateset}{S}

\newcommand{\stateset}{\expandafter\MakeUppercase\expandafter{\state}}
\newcommand{\Stateset}{\expandafter\MakeUppercase\expandafter{\State}}

\newcommand{\card}[1]{\left| {#1} \right|}
\newcommand{\Nat}{\ensuremath{\mathbb{N}}}
\newcommand{\Int}{\ensuremath{\mathbb{Z}}}
\newcommand{\Bool}{\ensuremath{\mathbb{B}}}

\newcommand{\smartpar}[1]{\medskip \noindent {\bf #1}}

\renewcommand{\iff}{\Leftrightarrow}

\iffinal
\renewcommand{\todo}[1]{}
\newcommand{\gray}[1]{}

\else
\newcommand{\gray}[1]{{\color{black!50} #1}}

\fi

\newcommand{\transition}[3]{#1 \xrightarrow{#2} #3}

\usetikzlibrary{arrows,automata,positioning, patterns, patterns.meta} %

\usetikzlibrary{positioning,arrows.meta,decorations.pathmorphing,calc,shapes.multipart,automata}

\tikzstyle{pta}=[auto, ->, >=stealth']
\tikzstyle{every node}=[initial text=]
\tikzstyle{location}=[rectangle, rounded corners, minimum size=12pt, draw=black, fill=blue!10, inner sep=3.5pt]
\tikzstyle{methodBox}=[rectangle, minimum size=12pt, draw=black, fill=green!10, inner sep=3.5pt]
\tikzstyle{invariant}=[draw=black, dotted, fill=yellow, inner sep=1pt, node distance=0] %
\tikzstyle{locguard}=[fill=cyan!30, inner sep=1pt, node distance=0] %
\tikzstyle{final}=[double, fill=blue!50]

\newcommand{\li}{\begin{itemize}}
\newcommand{\il}{\end{itemize}}

\NewEnviron{tikzLTS}{%
  \begin{tikzpicture}[node distance=2cm,inner sep=1pt,minimum size=0.5mm,bend angle=20]
	 	\tikzstyle{proc} = [rectangle,draw=black,fill=green!20,thick,inner sep=10pt]
	  	\tikzstyle{state} = [circle,draw=black,thick,inner sep=3pt]
	  	\tikzstyle{noproc} = [circle]
	  	\tikzstyle{lbl} = [rectangle,node distance=2cm]
  		\tikzstyle{pre} = [ <-,shorten <=2pt,shorten >=2pt, >=stealth', semithick]
	  	\tikzstyle{post} = [ ->,shorten <=2pt,shorten >=2pt, >=stealth', semithick]
    \BODY
  \end{tikzpicture}
}

\NewEnviron{tikzPath}{%
  \begin{tikzpicture}	
  	\tikzstyle{pstate} = [circle, fill=black,thick, inner sep=2pt, minimum size=2mm]
    \tikzstyle{hiddenstate} = [circle]
  	\tikzstyle{trans-notsched} = [ ->,shorten <=2pt,shorten >=2pt, >=stealth', semithick]
  	\tikzstyle{trans-sched} = [ ->,shorten <=2pt,shorten >=2pt, >=stealth', very thick] 
    \BODY
    \end{tikzpicture}
}

\newcommand{\ord}{\lesssim}
\newcommand{\zord}{\ord^z}
\usepackage{algorithm}
\usepackage[noend]{algpseudocode}
\algnewcommand{\LineComment}[1]{\State \(\triangleright\) #1}
\usepackage{float}
\newfloat{algorithm}{t}{lop}
\floatname{algorithm}{Algorithm}

\newcommand{\nuparrow}{\hspace{0.1cm}\uparrow \hspace{-0.1cm}}

%% file: intro.tex
\section{Introduction}
Due to the increasing prevalence and importance of distributed systems in our society, ensuring reliability and correctness of these systems has become paramount.
Computer-aided verification of distributed protocols and algorithms has been a very active research area in recent years \cite{HawblitzelHKLPR17,RahliGBC15,WilcoxWPTWEA15,McMillanP20,JaberWJKS21}.
To be practical, models of distributed algorithms need to take into account that communication or processes may be faulty, and correctness guarantees should be given based on a \emph{resilience condition} that defines the quantity and quality of faults (e.g., how many processes may crash or even behave arbitrarily~\cite{k17completeness,MaricSB17,Bertrand21}).

Moreover, many distributed systems consist of an arbitrary number of communicating processes, thus requiring \emph{parameterized verification} techniques that consider the number of processes as a parameter of the system, and that can provide correctness guarantees regardless of this parameter.
However, the parameterized verification problem is in general undecidable, even in restricted settings such as identical and anonymous finite-state processes that communicate by passing a binary-valued token in a ring~\cite{Suzuki88}. 

Since undecidability arises so easily in this setting, the research into automatic parameterized verification can in principle be divided into two directions: 
(i) identifying \emph{decidable} classes of systems and properties~\cite{German92,Emerso03,abdulla1996general,EsparzaFM99,EmersonK00,finkel2001well,EmersonK03,KaiserKW10,DelzannoSZ10,SchmitzS13,AminofJKR14,AminofKRSV14,BloemETAL15,AJK16,JS2018VMCAI,Balasubramanian18,JaberJW0S20,CzerwinskiO21,Balasubramanian22}, which are often restricted to rather specialized use cases, and 
(ii) developing \emph{semi-decision procedures}~\cite{Bouajjani00,AbdullaHH16}, which are usually much more versatile but come without a termination guarantee.
In practice however, the line between these two approaches is not so clear.
The line blurs when semi-decision procedures serve as decision procedures for a certain fragment of their possible inputs, and on the other hand, many decidable problems in parameterized verification have a huge complexity~\cite{SchmitzS13,CzerwinskiO21}, and to an engineer it does not make a difference if the verification algorithm is guaranteed to terminate within a thousand years, or whether it comes without a termination guarantee.

{In this paper, we aim for practical verification techniques that cover a large class of fault-tolerant distributed algorithms. 
Our focus is not decidability, but the development of techniques that can in practice verify a wide range of distributed algorithms, with full automation.}
We build on the formal model of threshold automata (TA) and existing techniques to verify them, which we combine with insights from the theory of well-structured transition systems (WSTS)~\cite{finkel2001well} and with abstraction techniques.
This combination allows us to lift some of the restrictions of the existing decidable fragments for TAs, and to verify distributed algorithms that are not supported by the existing techniques.

In particular, existing verification techniques for TAs are usually restricted to reason about shared variables that are monotonically increasing or decreasing over any run of the automaton.
An extension of the model that allows both increasing and decreasing variables has been considered before~\cite{kukovec2018reachability}, but only to show that the problem becomes undecidable, not to develop a verification technique for this case.
In our setting, we can allow increments, decrements and resets of shared variables, while still obtaining a semi-decision procedure, and even a decision procedure in certain cases.
In this paper, we demonstrate how this allows us to reason about round-based algorithms such as the phase king consensus algorithm~\cite{king} and the Red Belly Blockchain algorithm~\cite{CrainNG21}, which both use resets of shared variables at the beginning of a round, without a fixed bound on the number of rounds to be executed.

\smartpar{Contributions.}
In this paper,
\begin{enumerate}[nosep]
\item we present an extension of threshold automata that allows increments, decrements and resets of shared variables (\cref{sec:system-model}),
\item we develop a technique (based on well-structured transition systems) that is a decision procedure for general coverability properties of canonical threshold automata~\cite{kukovec2018reachability}, and is a semi-decision procedure for our
  extension (Sect.~\ref{sec:specs}-\ref{sec:checking-coverability}),
\item we develop an additional abstraction that reduces the search space of this technique and additionally allows us to check another type of specification, called reachability properties, on extended threshold automata (Sect.~\ref{sec:01-abstraction}),
\item we implement our techniques and demonstrate their performance on a number of examples from the literature,
  several of which --- including the \emph{phase king consensus} algorithm and the state-of-the-art Red Belly Blockchain algorithm ---
  cannot be modeled in canonical threshold automata (\cref{sec:experiments}).
\end{enumerate}

%% file: system-model.tex
\section{System Model} \label{sec:system-model}

In this section, we build on the existing notion of \emph{threshold automata} (TAs)~\cite{k17completeness} and generalize their notion of shared variables.
TAs are 
a model of distributed computation that encodes information exchange between processes into a fixed set of shared variables.
Shared variables in TAs are usually required to be monotonically increasing along an execution to ensure decidability.
We extend the definition such that it permits shared variables to be decreased or reset, which in general introduces undecidability \cite{kukovec2018reachability}.
We then define the semantics of an unbounded number of such TAs running in parallel.
In Section~\ref{sec:abstract-TA}, we introduce the notion of an abstract TA and provide the semantics for this abstraction.

\begin{definition}\label{def:TA}
A \emph{threshold automaton} is a tuple $A=(L, \mI, \Shared, \Pi,\Rules, RC)$ where:
\begin{itemize}
\item $L$ is a finite set of \emph{locations}.%
	\item $ \mI \subseteq L$ is the set of \emph{initial locations}.
	\item $\Shared = \{x_0,\ldots,x_m\}$ is a finite set of \emph{shared variables} over $\Nat_0$.
	\item $\Pi$ is a finite set of \emph{parameter variables} over $\Nat_0$.
	  Usually, %
          $\Pi = \{n,t,f\}$, where $n$ is the total number of processes, $t$ is a bound on the number of tolerated faulty processes
          , and $f$ is the actual number of faulty processes.
	\item $RC$, the \emph{resilience condition}, is a linear integer arithmetic formula over parameter variables.
E.g.: $RC = n > 3t\, \land\, t \geq f$.\\
For a vector $\vp \in \Nat_0^{|\Pi|}$, we write $\vp \models RC$ if $RC$ holds after substituting parameter variables with values according to $\vp$.
Then the set of \emph{admissible} parameters is $\vP_{RC} = \{\vp \in \Nat_0^{|\Pi|} : \vp \models RC\}$.
	
	\item $\Rules$ is a set of rules where a \emph{rule} is a tuple $r=(\from, \to, \varphi,\vup,\resets)$ such that:
	\begin{itemize}
		\item $\from,\to \in L$.
		\item $\vup \in |\Int|^{|\Shared|}$ is an \emph{update vector} for shared variables.
		\item $\varphi$ is a conjunction of lower guards and upper guards.
		A \emph{lower guard} has the form: $a_0 + \sum_{i=1}^{|\Pi|} a_i \cdot p_i \leq x$; An \emph{upper guard} has the form: $ a_0 + \sum_{i=1}^{|\Pi|} a_i \cdot p_i > x$, with $x \in \Shared$, $a_0,\ldots,a_{|\Pi|} \in \mathbb{Q}$, $p_1,\ldots,p_{|\Pi|} \in \Pi$.\\		
The left-hand side of a lower or upper guard is called a \emph{threshold}. 
	\item $\resets \subseteq \Shared$ is the set of shared variables to be reset to 0.
	\end{itemize}
\end{itemize}

\end{definition}

\input{ta_model}

\smartpar{Semantics of TA}.
Given a TA $A = (L, \mI, \Shared, \Pi,\Rules, RC)$, let the function $N: \: \vP_{RC} \rightarrow \Nat_0$ determine the number of processes to be modelled (usually, $N(n,t,f) = n-f$).
Then, the concrete semantics of a system composed of $N(\vp)$ threshold automata running in parallel are defined via a \emph{counter system}.

\begin{definition}\label{def:countersystem}
A \emph{counter system}~($\cs$) of a TA $A = (L, \mI, \Shared, \Pi,\Rules, RC)$ is a transition system $\csa = (\Confs, \Confs_0, \Trans)$ where%

\begin{itemize}
	\item $\Confs$ is the set of configurations.
	A \emph{configuration} is a tuple $\conf=(\vk,\vg,\vp)$ where:
	\begin{itemize}
		\item $\vk \in \Nat_0^{|L|}$ is a vector of counter values, where $\vk[i]$ represents the number of processes in location $i$. We refer to locations by their indices in $L$.
		\item $\vg \in \Nat_0^{|\Shared|}$ is a vector of shared variables' values, where $\vg[i]$ is the value of variable $x_i \in \Shared$.
		\item $\vp \in \vP_{RC}$ is an admissible vector of parameter values.
	\end{itemize}
	\item The set $\Confs_0$ contains all \emph{initial configurations}, i.e., configurations that satisfy
	\[ \forall x_i \in \Shared: \:\: \conf.\vg[i] = 0 \text{ and } \sum_{i \in \mI } \conf.\vk[i] = N(\vp) \text{ and } \sum_{i \not\in \mI } \conf.\vk[i] = 0 \]
	\item $\Trans \subseteq \Confs \times \Rules \times \Confs$ is the set of \emph{transitions}, where $(\conf, r, \conf') \in \Trans$ if and only if all of the following conditions hold:
	\begin{itemize}
		\item $\conf'.\vp = \conf.\vp$~(parameter values never change).
		\item $\conf'.\vk[r.\to] = \conf.\vk[r.\to] + 1$~(one process moves to $r.\to$).
		\item $\conf'.\vk[r.\from] = \conf.\vk[r.\from] - 1$~(one process moves out of $r.\from$)
		\item $\conf.\vg \models r.\varphi$ (i.e., $\varphi$ holds after replacing shared variables with values $\conf.\vg$) 
		\item $\conf'.\vg = \conf.\vg + r.\vup$
		\item $\forall x_i \in \resets$ $\conf'.\vg[i] = 0$
	\end{itemize}
	Instead of $(\conf, r, \conf') \in \Trans$ we also write $\transition{\conf}{r}{\conf'}$.
	If $(\conf, r, \conf') \in \Trans$, we say $r$ is \emph{enabled} in $\conf$; otherwise it is \emph{disabled}.
\end{itemize}
\end{definition}

\smartpar{Paths of $\cs$.} A sequence $\conf_0, r_0, \conf_1,\ldots,\conf_{k-1},r_{k-1}, \conf_k$ of alternating configurations and rules is a \emph{path} of a counter system $\csa = (\Confs, \Confs_0, \Trans)$ if and only if $\conf_0 \in \Confs_0$ and $(\conf_{i},r_i,\conf_{i+1}) \in \Trans$  for $0 \leq i < k$.
In this case we also write $\conf_0 \rightarrow^* \conf_k$.
We denote by $Paths(\csa)$ the set of all paths of $\csa$.

\begin{example}
	Let $N(n,t,f) = n-f$, $RC = n > 3t \land t \geq f$, then the following is a valid path of the counter system of the TA in Figure \ref{fig:ta}:\\
	$[(4,0,0,0,0)(0,0)],r_0,[(3,0,1,0,0)(1,0)],r_0,[(2,0,2,0,0)(2,0)],r_0$,\\
	$[(1,0,3,0,0)(3,0)],r_0,[(0,0,4,0,0)(4,0)],r_2,[(0,0,3,1,0)(4,0)]$.
\end{example}
 
\subsection{Abstract Threshold Automata}
\label{sec:abstract-TA}

The shared variables and parameters of a TA have infinite domains.
To facilitate the application of parameterized verification techniques for finite-state processes, we introduce an abstraction of TAs based on parametric interval abstraction~\cite{john2013parameterized}.

\smartpar{Abstract Domain}.
Given a TA $A$, define 
as $\mTH=\{d_0,d_1,\ldots,d_k\}$ the set of thresholds where $d_0=0,d_1=1$ and $\forall i >1$ $d_i$ is a threshold in $A$.
We assume that $\forall i,j$ $d_i < d_j$ if $i < j$.
Note that this is always possible for a fixed $\vp \in \vP_{RC}$.
If different $\vp \in \vP_{RC}$ result in different orders of the $d_i$, then we consider each of the finitely many such orders separately.
Based on this, define the finite set of intervals $\mD=\{I_0,I_1,\ldots,I_k\}$ where $I_i = [d_i,d_{i+1}[$ if $i<k$, and $I_k = [d_k,\infty[$.

\begin{definition}\label{def:abstractTA}
\smartpar{Abstract Threshold Automata.}
Given a threshold automaton $A=(L, \mI, \Shared, \Pi,\Rules, RC)$, we define the \emph{abstract threshold automaton}~(or $\absTA$) $\absA=(L, \mI, \absShared, \Pi,\absRules)$ where:
\begin{itemize}
	\item $A$ and $\absA$ share the components $L,\mI,\Pi$.
	\item Let $\Shared = \{x_0,\ldots,x_m\}$, then $\absShared= \{\bx_0,\ldots,\bx_m\}$, where each $\bx_i$ is over the domain $\mD=\{I_0,I_1,\ldots,I_k\}$.
	\item $\absRules$ is the set of abstract rules. An \emph{abstract rule} is a tuple $\ar=(\from, \to, \bvarphi, \vup,$ $ \resets)$ where $\from, \to, \vup, \resets$ are as before, and the \emph{abstract guard} $\bvarphi$ is a Boolean expression over equalities between shared variables and abstract values.
          
Formally, let $\varphi = \varphi_0 \land \ldots \land \varphi_k$, then $\bvarphi = \bphi_0 \land \ldots \land \bphi_k$ where for $\varphi_i = (d_j \leq x)$, we have $\bphi_i = \bigvee_{c=j}^{k-1} (\bx = [d_c,d_{c+1}[) \lor \bx = [d_k,\infty[$, and for $\varphi_i = (d_j > x)$, we have $\bphi_i = \bigvee_{c=0}^{j-1} (\bx = [d_c,d_{c+1}[)$.

\end{itemize}
\end{definition}

\begin{example}
Consider again the TA in Figure \ref{fig:ta} with $N(n,t,f) = n-f$, $RC = n > 3t \land t \geq f > 1$.
We have $\mTH=\{0,1,t,n-t\}$ and $\mD=\{[0,1[,[1,t[,[t,n-t[, [n-t,\infty[\}$ whose order is induced by $RC$.
Moreover, we have $\ar_0 = r_0,\ar_1 = r_1$~(due to the absence of a guard), $\ar_2.\bvarphi = (\bx_0 = [n-t,\infty[)$, $\ar_3.\bvarphi = (\bx_1 = [n-t,\infty[)$.
\end{example}

To keep the presentation simple, in our definition all shared variables have the same abstract domain. 
The abstraction can be improved by considering different abstract domains for different variables:
for a given shared variable $x$ we can let $\mTH_x=\{d_0,d_1,\ldots,d_l\}$ where $d_0=0,d_1=1$ and $\forall i >1$ there is a guard $d_i * x$ with $* \in \{\geq, <\}$, to obtain a corresponding abstract domain $\mD_x$ for $x$.

\subsubsection{Semantics of $\absTA$.}
We first over-approximate the semantics of a system composed of an arbitrary number of $\absTA$s by an abstract counter system. 
We later show how to detect whether a behavior of the abstract counter system corresponds to a concrete behavior of a counter system.

\begin{definition}\label{def:abstractCS}
An \emph{abstract counter system}~($\acs$) of $\absA=(L, \mI, \absShared, \Pi,\absRules)$ is a transition system $\acsa = \abcs$ where:
\begin{itemize}
	\item $\absConfs$ is the set of abstract configurations.
	A \emph{configuration} of $\acsa$ is a tuple $\absconf=(\vk,\avg)$ where:
	\begin{itemize}
		\item $\vk \in \Nat_0^{|L|}$ is a vector of counter values where $\vk[i]$ represents the number of processes in location $i$.
		\item $\avg \in \mD^{|\absShared|}$ is a vector of shared variables values, where $\avg[i]$ is the parametric interval currently assigned to $\bx_i$.		
	\end{itemize}
	\item The set $\absConfs_0$ contains all \emph{initial abstract configurations}, i.e., those that satisfy 
	\[ \forall \bx_i \in \absShared: \:\: \absconf.\avg[i] = I_0 \text{ and } \sum_{i \in \mI } \absconf.\vk[i] \geq 0 \text{ and } \sum_{i \not\in \mI } \absconf.\vk[i] = 0 \]
	\item $\absTrans \subseteq \absConfs \times \absRules \times \absConfs$ is the set of transitions. 
	A \emph{transition} is a tuple $t=(\absconf, \ar, \absconf')$ where:	
	\begin{itemize}
		\item $\absconf'.\vk[\ar.\to] = \absconf.\vk[\ar.\to] + 1$
		\item $\absconf'.\vk[\ar.\from] = \absconf.\vk[\ar.\from] - 1$
		\item $\absconf.\avg \models \ar.\bphi$
		\item $\absconf'.\avg = \absconf.\avg \dot{+} \ar.\vup$, defined as follows: $\forall i < |\Shared| $:
		\begin{enumerate}
			\item $\absconf'.\avg[i] = \absconf.\avg[i]$, if $\ar.\vup[i] = 0$
			\item $(\absconf'.\avg[i] = \absconf.\avg[i]) \lor (\absconf'.\avg[i] = \absconf.\avg[i].next)$, if $\ar.\vup[i] = 1$
			\item $(\absconf'.\avg[i] = \absconf.\avg[i]) \lor (\absconf'.\avg[i] = \absconf.\avg[i].previous)$, if $\ar.\vup[i] = -1$
			\item[]the first disjunct in 2 and the second in 3 are omitted if $\absconf.\avg[i] = I_0$.
			\item $\forall x_i \in \ar.\resets: \:\: \absconf'.\avg[i] = I_0$
		\end{enumerate}
	\end{itemize}
	We also write $\transition{\absconf}{\ar}{\absconf'}$ instead of $(\absconf, \ar, \absconf') \in \absTrans$.
\end{itemize}
\end{definition}

\noindent
In contrast to prior work~\cite{john2013parameterized}, the domain of counters is not abstracted in $\acs$.

\smartpar{Paths of $\acs$.}
A sequence $\absconf_0, \ar_0, \absconf_1,\ldots,\absconf_{k-1},\ar_{k-1}, \absconf_k$ of alternating abstract configurations and rules is called a \emph{path} of $\acsa = \abcs$, if $\absconf_0 \in \absConfs_0$ and  $(\absconf_{i},\ar_i,\absconf_{i+1}) \in \absTrans$  for $0 \leq i \leq k$.
In this case we also write $\absconf_0 \rightarrow^* \absconf_k$.
We denote by $Paths(\acsa)$ the set of all paths of $\acsa$.

\begin{example}
\label{ex:path}
  Let $I_0=[0,1[,I_1=[1,t[,I_2=[t,n-t[,I_3=[n-t,\infty[ $.
	The following is a valid path of the abstract counter system of the TA in Figure \ref{fig:ta}:
	$[(4,0,0,0,0)(I_0,I_0)],\ar_0,[(3,0,1,0,0)(I_1,I_0)],\ar_0,[(2,0,2,0,0)(I_2,I_0)],\ar_0$,\\$[(1,0,3,0,0)(I_3,I_0)],\ar_2,[(1,0,2,1,0)(I_3,I_0)]$.
\end{example}

\smartpar{Relation between $\acsa$ and $\csa$.}
In comparison to $\cs$, in $\acs$ we drop the resilience condition, as well as the function $N$ that determines the number of processes to be modeled.
Moreover, a transition in $\acs$ may jump from one interval to the next too early and may stay in the same interval although it had to move.
We will formalize the relation between the two models in Sect.~\ref{sec:spurious}.

%% file: ta_model.tex
\vspace*{1em}
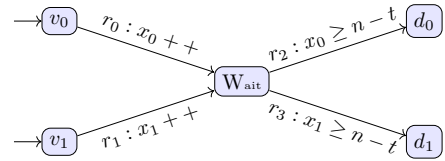
\begin{wrapfigure}[7]{r}{.49\textwidth}
	\vspace{-7em}
	\scalebox{0.8}{
		\begin{tikzpicture}		
			\node [location, initial, initial text=] (v0) at (0,0) {$v_0$};
			\node [location] (w)  at (3, -1) {W\tiny{ait}};
			\node [location, initial, initial text=] (v1) at (0,-2)  {$v_1$};
			\node [location] (d0) at (6,0)  {$d_0$};
			\node [location] (d1)  at (6,-2)  {$d_1$};

			\path[->]
				(v0) edge [] node [sloped, above] {$r_0: x_{0}++$} (w)
				(v1) edge [] node [sloped, below] {$r_1: x_{1}++$} (w)
				(w) edge [] node [sloped, above] {$r_2: x_{0} \geq n-t$} (d0)
				(w) edge [] node [sloped, below] {$r_3: x_{1} \geq n-t$} (d1)
			;
		\end{tikzpicture}
	}

\caption{A threshold automaton for simple voting.}
\label{fig:ta}
\end{wrapfigure}

\noindent
\begin{minipage}{.47\textwidth}
	\begin{example}
		\label{ex:TA}
		Figure~\ref{fig:ta} sketches a threshold automaton with $\mI=\{v_0,v_1\}$, $L = \{v_0,v_1, Wait, d_0,d_1\}$, $\Shared=\{x_0,x_1\}$, $\Pi=\{n,t,f\}$. 
		A process in $v_0$ has a vote of $0$ and a process in $v_1$ has a vote of $1$. 
		If at least $n-t$ processes vote with $0$~(respectively, $1$), the decision will be $0$~($1$), modeled by all processes moving to $d_0$ ($d_1$).
	\end{example}		
\end{minipage}
\vspace*{1em}

%% file: specs.tex
\section{Specifications}
\label{sec:specs}
We consider three kinds of specifications: \emph{general coverability}, which refers to the notion of coverability that is widely used in parameterized verification, e.g., Petri nets~\cite{Finkel91} or VASS~\cite{rackoff1978covering}, (non-general) \emph{coverability} and \emph{reachability}. The latter two
are specifications used in prior work on threshold automata~\cite{esparza2020TAcomplexity,konnov20172,kukovec2018reachability}.
Note that general coverability can express mutual exclusion, e.g., that there cannot be two leaders at the same time, while coverability cannot.

\begin{definition}\label{def:gen-coverability}
	The \emph{general parameterized coverability problem} is: Given $\csa$ and a \emph{general coverability specification} $\Confs_{spec} \subseteq \Confs$, decide if there is a path in $\csa$ that covers $\Confs_{spec}$ (i.e., decide if there is some configuration $\conf_r \in \Confs$ reachable from $\conf_0 \in \Confs_0$ and  $\exists \conf_{spec} \in \Confs_{spec}$ where $\forall i \: \conf_{spec}.\vk[i] \leq \conf_r.\vk[i]$).
\end{definition}

\begin{definition}\label{def:coverability}
	The \emph{parameterized coverability problem} is: Given a TA and \emph{coverability specification} $L_{spec} = L_{>0}$, decide if there is a path of its $\csa$ that satisfies $L_{spec}$ (i.e., decide if there is some configuration $\conf_r \in \Confs$ reachable from $\conf_0 \in \Confs_0$ and  $\conf_r$ satisfies $\forall i \in L_{>0} \: \conf_r.\vk[i] > 0$).
\end{definition}

\begin{definition}\label{def:reachability}
	The \emph{parameterized reachability problem} is: Given a TA and a \emph{reachability specification} $L_{spec} = (L_{=0}, L_{>0})$, decide if there is a path of $\csa$ that satisfies $L_{spec}$ (i.e., decide if there is some configuration $\conf_r \in \Confs$ reachable from $\conf_0 \in \Confs_0$ and  $\conf_r$ satisfies $\forall i \in L_{=0} \: \conf_r.\vk[i] = 0$ and  $\forall i \in L_{>0} \: \conf_r.\vk[i] > 0$).
\end{definition}

We define similarly all three types of problems for an abstract TA and its abstract counter system.
Usually, our specifications are definitions of error configurations, and therefore paths that satisfy them, are called \emph{error paths}.

%% file: spurious.tex
\section{$\cs$ vs $\acs$}\label{sec:spurious}

We now show that the abstraction from $\cs$ to $\acs$ is complete with respect to the specifications introduced in the previous section.

A path  $\bpi = \absconf_0, \ar_0,\ldots,\ar_{m-1}, \absconf_m$ in $\acsa = \abcs$ \emph{corresponds} to the paths $\pi = \conf_0, r_0^{c_0},\ldots,r_{m-1}^{c_{m-1}}, \conf_m$~(where $r_i^{c_i}$ simulates applying $r_i$ $c_i$ times)  of $\csa=(\Confs, \Confs_0, \Trans)$ that satisfy the following conditions:\begin{itemize}

	\item $
	RC \land (\textstyle\sum_{j \in \mI } \conf_0.\vk[j] = N(n,t,f))  $
	\item $\forall i < m \; \conf_{i}.\vk[r_i.\from] = c_i + \conf_{i+1}.\vk[r_i.\from] \land
	\conf_{i+1}.\vk[r_i.\to] = c_i + \conf_{i}.\vk[r_i.\to] $
	\item $ \forall i < m \; \forall x_j \in \Shared \;
	x_j \notin r_i.\tau \implies
	\conf_{i+1}.\vg[j] = \conf_{i}.\vg[j] + c_i \cdot r_i.\vup[j]$
	\item $\forall i < m \; \forall x_j \in r_i.\resets \; \conf_{i+1}.\vg[j] = 0$
	\item $\forall i < m \;
	\forall x_j \in \Shared \;  \conf_{i}.\vg[j] \in \absconf_{i}.\avg[j] \land \conf_{i+1}.\vg[j] \in \absconf_{i+1}.\avg[j]
	$
	\item $\forall i < m \; c_i > 1 \implies ((\conf_{i+1}.\vg - r_i.\vup) \models r_i.\varphi)$\footnote{This is needed only in cases where an update affects any of the guards of $r_i.\varphi$}

\end{itemize}
Let $Concretize(\bpi)$ be the conjunction of the constraints above, where quantified formulas are instantiated to a finite conjunction of quantifier-free formulas.
Note that $Concretize(\bpi)$ is a quantifier-free formula in linear integer arithmetic, and a satisfying assignment for $Concretize(\bpi)$ (that can be computed by an SMT solver) represents a path of $\csa$ that corresponds to $\bpi$. 
We say that a path $\bpi \in Paths(\acsa)$ is \emph{spurious} if $Concretize(\bpi)$ is unsatisfiable. 

For a given $\bpi \in Paths(\acsa)$, let $Cover(\bpi)= \forall l \in L \; \conf_m.\vk[l] \geq \absconf_m.\vk[l]$.
We can show the
following
connection between error paths in $\csa$  and $\acsa$:

\begin{restatable}{lemma}{lemGenCovPathequiv}\label{lem:gen-cov-pathequiv}
$\acsa$ has a non-spurious path that covers a set of configurations $\Confs_{spec} \subseteq \Confs$ iff $\csa$ has a path that covers $\Confs_{spec}$.
\end{restatable}

\noindent Note that Lemma \ref{lem:gen-cov-pathequiv} subsumes  the case of non-general coverability.

Similarly, if $\bpi$ is a non-spurious path of $\acsa$ that satisfies a reachability specification $L_{spec}$, let $Reach(\bpi)= \forall l \in L \; \conf_m.\vk[l] > 0 \iff \absconf_m.\vk[l] > 0$.
Then we can show the following with a similar proof as above, where we replace $Cover(\bpi)$ with $Reach(\bpi)$ and reason accordingly.

\begin{lemma}\label{lem:reach-cov-pathequiv}
$\acsa$ has a non-spurious path that satisfies a reachability specification $L_{spec}$ iff $\csa$ has a path that satisfies $L_{spec}$.
\end{lemma}

%% file: Thresh-Autom-WSTS.tex
\section{Checking General Parameterized Coverability}
\label{sec:checking-coverability}

In this section, we show how abstract counter systems (including ours) can be framed as well-structured transition systems. We also introduce a parameterized model checking algorithm for checking general parameterized coverability in $\acs$.

\input{wsts}

\subsection{WSTS-based
  General Parameterized Coverability Checking}

In this section, we present our algorithm for solving the general parameterized coverability problem~(see Definition \ref{def:gen-coverability}).
Given an abstract counter system $\acsa$ and a finite set of error configurations $ERR \subseteq \absConfs$, 
we say that a path of $\acsa$ is an \emph{error path} if it starts in $\absConfs_0$ and ends in the upward closure of $ERR$.
Lemma \ref{lem:predbasis} enables us to use the transitive closure of $CPredBasis(ERR)$ to compute a set of error paths in $\acsa$ or a fixed-point in which no initial configuration occurs.
As defined in Section \ref{sec:spurious}, we can examine whether any of these abstract error paths corresponds to error paths in $\csa$.
The detailed algorithm is given in Algorithm~\ref{alg:ParamCC}, which we explain in the following.

\input{algorithm}

\subsection{Correctness.}

Algorithm \ref{alg:ParamCC} is sound, complete, and terminates when the TA is restricted as in \cite{esparza2020TAcomplexity,konnov2017short,kukovec2018reachability, konnov20172, konnov2018bymc}.
Soundness follows directly from encoding non-spuriousness into the constraint $Concretize(\bpi) \land Cover(\bpi)$, as defined in Section \ref{sec:spurious}.

\begin{corollary}[Soundness]
	Algorithm \ref{alg:ParamCC} is sound. That is, if the algorithm computes a non-spurious error path, then there is a configuration in $\nuparrow ERR$ that is reachable in $\csa$.
\end{corollary}

With Lemma~\ref{lem:gen-cov-pathequiv}, the following lemma proves the algorithm's completeness.

\begin{restatable}[Completeness]{lemma}{lemConstructCE}\label{lem:ConstructCE}
If $\acsa$ has a non-spurious error path $\bpi$, then Algorithm \ref{alg:ParamCC} will find it.
\end{restatable}

\smartpar{Termination}. Since Algorithm \ref{alg:ParamCC} is sound and complete, it will terminate whenever $\cs$ has a path that covers $\Confs_{spec}$.
However, Lines \ref{line:updateMap1} and \ref{line:updateMap2} of the algorithm could create cycles in $errGraph$. 
In the presence of cycles however, the algorithm could compute an infinite sequence of error paths.
Therefore, we proved the following theorem.

\begin{restatable}{theorem}{thParamCCTerm}\label{th:ParamCC-Term}
If $\forall \ar \in \absRules: \ar.\vup \in |\Nat_0|^{|\Shared|} \land \ar.\resets = \emptyset$, then Algorithm \ref{alg:ParamCC} terminates.
\end{restatable}

It is important to note that if $errGraph$ is acyclic, the algorithm is also guaranteed to terminate.
This includes the case when $errGraph$ is empty, i.e., the corresponding $\acs$ has neither spurious nor non-spurious error paths.
The algorithm also terminates if no cycle in $errGraph$ has
decrements or resets.

%% file: wsts.tex
\subsection{Well-structured transition systems}

\emph{Well-structured transition systems}~\cite{finkel2001well} (WSTS) are a class of infinite-state systems for which the general parameterized coverability problem is decidable \cite{abdulla1996general,jacobs2022automatic}.
In the following, we recap the standard definitions of WSTS.

\smartpar{Well-quasi-order.}
Given a set $\cstateset$, a binary relation $\preceq \: \subseteq \cstateset \times \cstateset$ is a \emph{well-quasi-order}
(wqo) if $\preceq$
is reflexive, transitive, and if any infinite sequence $\cstate_0,\cstate_1,\ldots
\in \cstateset^{\omega}$ contains a pair $\cstate_i \preceq \cstate_j$ with $i < j$.
A subset $R \subseteq \cstateset$ is an \emph{anti-chain} if any two distinct elements of $R$ are incomparable wrt. $\preceq$.
Therefore, $\preceq$  is a wqo on $\cstateset$ if and only if it is reflexive, transitive, and has no infinite anti-chains.
The \emph{upward closure} of a set $R \subseteq \cstateset$,
denoted ${\uparrow}R$, is the set $\{\cstate \in \cstateset \mid \exists \cstate' \in R: \cstate' \preceq
\cstate\}$.
We say that $R$ is \emph{upward-closed} if ${\uparrow}R = R$, and we call $B
\subseteq S$ a \emph{basis} of $R$ if ${\uparrow}B = R$. 
If $\preceq$ is also anti-symmetric, then any basis of $R$ has a unique subset of minimal elements.
We call this set the \emph{minimal basis} of $R$, denoted $minBasis(R)$.

\smartpar{Compatibility.}
Given a transition system $M=(\cstateset,\cstateset_0,\Trans)$, we say that a wqo $\preceq \; \subseteq \cstateset \times \cstateset$ is \emph{compatible} with the transition relation $\Trans$ if the following
holds:
$$ \forall \cstate,\cstate',\cstate_x \in \cstateset: \text{ if } \transition{\cstate}{}{\cstate'}
\text { and } \cstate \preceq \cstate_x \text{ then }
\exists \cstate_x' \text{ with } \cstate' \preceq \cstate_x'
\text{ and } \transition{\cstate_x}{}{^*\cstate_x'},$$
where $\transition{\cstate}{}{\cstate'}$ is a transition in $\Trans$, and $\transition{\cstate_x}{}{^*\cstate_x'}$ is a path in $M$.

\smartpar{WSTS.}
We say that $(M, \preceq)$ with $M=(\cstateset,\cstateset_0,\Trans)$ is a \emph{well-structured transition system} if $\preceq$ is a wqo on $\cstateset$ that is compatible with $\Trans$.
The set of \emph{immediate predecessors} of a set $R\subseteq \cstateset$ is
$pred(R) = \{\cstate \in \cstateset \mid \exists \cstate' \in R: \transition{\cstate}{}{\cstate'} \}.$
We say that a WSTS  $(M,\preceq)$ \emph{has effective $pred$-basis} if there
exists an algorithm that takes as input any finite set $R \subseteq \cstateset$ and
returns a
finite basis of $pred({\uparrow}R)$.

\subsection{Abstract Counter Systems as WSTS}\label{sec:acsWSTS}

To prove that the general parameterized coverability problem is decidable for abstract TAs, %
we show that an $\acs$ can be framed as WSTS.
Here, for a given set $E' \subseteq \absConfs$ we define $pred(E') = \{ \absconf \in \absConfs \mid \absconf' \in E' \land \exists \ar \in \absRules \text{ s.t. } (\absconf, \ar, \absconf') \in \absTrans\}$.

\begin{restatable}{lemma}{lemAbcsWSTS}\label{lem:abcsWSTS}
	Given an abstract counter system $\acsa =  \abcs$ let
	$\ord \; \subseteq \absConfs
	\times \absConfs$ be the binary relation defined by:
	$$(\vk,\avg) \ord (\vk',\avg') \iff \vk \leq \vk' \land \avg = \avg'$$
	where $\leq$ is the component-wise ordering of vectors.
	Then $(\acsa,\lessapprox)$ is a WSTS.
\end{restatable}

Note also that the order $\ord$ is anti-symmetric, and therefore every upward-closed set of configurations has a unique minimal basis.

\begin{restatable}{lemma}{lemPredbasis}\label{lem:predbasis}
	Given an abstract counter system $\acsa = \abcs$, the WSTS $(\abcs, \ord)$ has effective pred-basis.	
\end{restatable}

Let $BasisTrans$ be the transitions from which we computed $CPredBasis$ in the proof of Lemma \ref{lem:predbasis}. Concretely, let $\vu_j$ be the unit vector with $\vu_j(j) = 1$ and $\vu_j(i) = 0$ for $i \neq j$ then $BasisTrans(E')$ is the set
\[
\left\{
\begin{array}{l|l}
((\vk,\avg), \ar,(\vk',\avg')) \in \absTrans 
&	(\vk',\avg') \in E' \land \ar= (l_i,l_j,\bphi,\vup, \resets) \land \\ 
& \left[ \transition{(\vk,\avg)}{\ar}{(\vk',\avg')} \lor   (\transition{(\vk,\avg)}{\ar}{(\vk' {+} \vu_j,\avg')}{\land}\vk'[j]{=}0) \right] 
									\end{array}					
											\right\}.
\]

The following corollary, derived from the aforementioned definitions,  will be instrumental in establishing the correctness of our algorithms.

\begin{corollary}\label{cor:compStatesPred}
	Given $\absconf'_1,\absconf'_2 \in \absConfs$, if $\absconf'_1 \ord \absconf'_2$ then $\{\ar \in \absRules \mid (\absconf_1,\ar,\absconf'_1) \in BasisTrans(\absconf'_1) \} = \{\ar \in \absRules \mid (\absconf_2,\ar,\absconf'_2) \in BasisTrans(\absconf'_2) \}$.
\end{corollary}

%% file: algorithm.tex
\begin{algorithm}[h]
	\caption{General Parameterized Coverability Checking}\label{alg:ParamCC}
	
	\begin{algorithmic}[1]		
		\Procedure{CheckCoverability}{\emph{Abstract Counter System} $ACS$,$ERR$}
		\State $E_0 \gets ERR, i \gets 1$, $errGraph \gets \emptyset$
		\State $visitedTrans \gets \emptyset$ \label{line:mc-initialise}{\color{brown}//set of visited transitions} 
		\While{$E_{i-1} \neq \emptyset$} \label{line:mc-while}  {\color{brown}//has a fixed-point been reached?} 
		\State $E_i, visitedTrans \gets \Call{ComputePredBasis}{E_{i-1},visitedTrans}$ \label{line:mc-pred}		
		\State $i \gets i+1$ 
		\EndWhile
		\State $visitedInitialConfigs \gets \bigcup_{j <i} (E_j \cap \Confs_0)$\label{line:mc-intersect0}
		\If{$visitedInitialConfigs \neq \emptyset$}\label{line:mc-intersect} \ \ {\color{brown} //intersects with initial configurations?}
		\State $errGraph \gets \Call{ConstructErrGraph}{visitedInitialConfigs,visitedTrans}$\label{line:errTree}
		\State $nonSpuriousCE = CheckforNonSpuriousCEs(errGraph)$\label{line:smt}
		\If{$nonSpuriousCE \neq \emptyset$} {\color{brown} //at least one CE is non-spurious}
		\State \textbf{return} $nonSpuriousCE$ {\color{brown} //an error is found!}\label{line:mc-error}
		\EndIf
		\EndIf
		\State \textbf{return} `` The system is safe! ''	\label{line:mc-true}			
		\EndProcedure
		
		\Procedure{ComputePredBasis}{$E_{i-1}$,$visitedTrans$}
		\State Compute $BasisTrans(E_{i-1})$ as explained at the end of Section \ref{sec:acsWSTS}
		\State $E_i= \{\absconf \in \absConfs \mid \exists (\absconf,\ar,\absconf') \in BasisTrans(E_{i-1})\}$\label{line:lemma}%
		\State $visitedTrans.add(BasisTrans(E_{i-1}))$ \label{line:appendAll}
		\State $finalE_i \gets \emptyset$
		\ForAll{$\absconf \in E_i$} {\color{brown} //remove visited bigger configurations to ensure termination}
		\If{$\exists j \leq i, \absconf_s \in E_j$ s.t. $\absconf_s \ord \absconf$} {\color{brown} //smaller configurations, maybe $\geq$ 1}
		\State $visitedTrans.add(\{(\absconf_s,\ar,\absconf') \mid (\absconf,\ar,\absconf') \in  BasisTrans(E_{i-1}) \})$\label{line:updateMap1}
		\Else
		\State $finalE_i.add(\absconf)$
		\If{$\exists \absconf_b \in finalE_i$ s.t. $\absconf \ord \absconf_b$ } {\color{brown} //bigger configurations, maybe $>1$}
		\State $visitedTrans.add(\{(\absconf,\ar,\absconf') \mid (\absconf_b,\ar,\absconf') \in  BasisTrans(E_{i-1}) \})$\label{line:updateMap2}
		\State $finalE_i.remove(\absconf_b)$
		\EndIf
		\EndIf
		\EndFor
		\State \textbf{return} $finalE_i,visitedTrans$  
		\EndProcedure
	\end{algorithmic}	
\end{algorithm}

Procedure $\textproc{CheckCoverability}$ takes as argument an $\acsa$ and a basis for a set of error configurations $ERR$.
After initializing local variables, the procedure enters a while loop that, given $E_{i-1}$,
invokes
$\textproc{ComputePredBasis}$~(Line \ref{line:mc-pred}), a sub-procedure to compute $E_i \supseteq minBasis(pred({\uparrow}E_{i-1}))$ such that $\forall \absconf \in {\uparrow}E_i \:\; \forall j <i \:\; \absconf \not \in {\uparrow}E_j$.
$\textproc{ComputePredBasis}$ also updates $visitedTrans$ which represents the set of transitions
explored so far.
The loop breaks once a fixed-point has been reached.
Line \ref{line:mc-intersect0} computes the visited set of initial configurations $visitedInitialConfigs$.
If $visitedInitialConfigs$ is not empty then the procedure extracts all computed error paths~(Line \ref{line:errTree}).
$\textproc{ConstructErrGraph}$ starts from the discovered initial configuration(s) and  uses $visitedTrans$ to construct the error graph $errGraph$ which encodes all error paths. 
In a breadth-first fashion, the procedure $\textproc{CheckforNonSpuriousCEs}$ then unfolds $errGraph$ and evaluates the spuriousness of uncovered error paths that starts in $\absConfs_0$ and ends in $ERR$~(Line \ref{line:smt}).
A path $\bpi=\absconf_0,\ar_0,\ldots,\ar_{j-1},\absconf_j$ in $errGraph$ is spurious if $Concretize(\bpi) \land Cover(\bpi)$ is unsatisfiable (see Sect.~\ref{sec:spurious}).
In the presence of cycles\footnote{A \emph{cycle} is a sub-sequence of a path that starts and ends in the same configuration.}, unfolding may not terminate in the presence of shared-variable decrements and/or resets~(see Theorem \ref{th:ParamCC-Term}).
If all error paths are spurious, we conclude that the system is safe. Otherwise a concrete path is returned as a witness of the buggy system~(Line \ref{line:mc-error}).

Procedure $\textproc{ComputePredBasis}$ takes as argument the current set $E_{i-1}$, and the set of visited transitions $visitedTrans$.
We compute $E_i$ as the predecessors of $E_{i-1}$ as explained at the end of Section \ref{sec:acsWSTS}~(Line~\ref{line:lemma}).
At this point, $E_i$ may contain configurations larger than those we have already explored.
To ensure termination, they must be removed.
However, before removing these configurations, we collect their visited transitions since these may introduce new, unexplored behaviors.
The computation of a comprehensive set of visited transitions, denoted as $visitedTrans$ in this context, is crucial for the correctness of our algorithm.
The importance of this lies in the necessity to address spurious counterexamples, which mandates retrieving all paths from an initial configuration to an error configuration. This is in contrast to ordinary backward model checking, where a single counterexample suffices.
After adding all computed transitions to $visitedTrans$~(Line \ref{line:appendAll}), we check for each configuration $\absconf$ in $E_i$ the following:
\begin{itemize}
\item For every configuration $\absconf_s$ with $\absconf_s \ord \absconf$, we add
    from $BasisTrans(E_{i-1})$ transitions that start in $\absconf$ to $visitedTrans$ after replacing $\absconf$ with $\absconf_s$.
	\item Otherwise, we add $\absconf$ to $finalE_i$.
	Also, for every $\absconf_b \in finalE_i$ with $\absconf \ord \absconf_b$, we add from $BasisTrans(E_{i-1})$ transitions that start in $\absconf_b$ to $visitedTrans$ after replacing $\absconf_b$ with $\absconf$, and we remove $\absconf_b$ from $finalE_i$.
\end{itemize}

%% file: zcs.tex
\section{Reachability via $(0,1)$-Abstraction}
\label{sec:01-abstraction}

Two configurations $\conf, \conf'$ may be comparable with respect to order $\ord$ even if some location $l$ is occupied in $\conf$ while it is not occupied in $\conf'$.
This implies that an algorithm based on upward-closed sets wrt. $\ord$ cannot be used to decide reachability properties.
Furthermore, note that coverability and reachability specifications (Def.~\ref{def:coverability} and \ref{def:reachability}) are agnostic to the precise number of processes in each location, i.e., they require only to distinguish between locations that are occupied by one or more processes and those that are not occupied.

To enable reachability checking and to enhance the performance of our algorithm for (non-general) coverability checking, we introduce a second, similar abstraction of our system model as in \cite{jacobs2022automatic}, where each counter can only assume one of two values:
$1$ to indicate that the location is currently occupied by at least one process; and $0$ to indicate that it is not occupied.

\smartpar{$(0,1)$-configuration.}
Given an abstract threshold automaton  $\absA=(L, \mathcal{I}, \absShared, \Pi,$ $\absRules)$, a \emph{$(0,1)$-configuration} is a tuple $\zconf = (\vkz, \avg)$, where $\vkz \in \Bool^{\card{L}}$, and $\avg$ is defined as before.
That is $\vkz[i]$ indicates the presence ($1$) or absence ($0$) of at least one process at location $i$.

\begin{definition}
A \emph{$(0,1)$-counter system} (or $\zcs$) of $\absA=(L, \mathcal{I}, \absShared, \Pi,\absRules)$, is a \emph{transition system} $\zcsa = (\zConfs, \zConfs_0, \mTz)$, where:
\begin{itemize}
	\item $\zConfs = \Bool^{\card{L}} \times \mD^{|\absShared|}$, is the set of $(0,1)$-configurations
	\item $\zConfs_0 \subseteq \zConfs$ is the set of $(0,1)$-configurations $\zconf$ that satisfy the following:
	\begin{itemize}
		\item $\forall i \in \Shared: \:\: \zconf.\avg[i] = I_0$
		\item $\forall i \in L: \:\: \zconf.\vkz[i]=1 \implies i \in \mathcal{I}$
	\end{itemize}
	\item the transition relation $\mTz$ is the set of transitions $(\zconf, \ar, \zconf')$ with:
	
	\begin{itemize}
		\item $\ar=\{\from,\to,\bvarphi,\vup\} \in \absRules$
		\item $\zconf.\avg \models \ar.\bvarphi$
		\item $\zconf.\vkz[\ar.\from] = 1$ and $(\zconf'.\vkz[\ar.\from] = 0$ or $\zconf'.\vkz[\ar.\from] = 1)$
		\item $\zconf'.\vkz[\ar.\to] = 1$
		\item $\zconf'.\avg = \zconf.\avg \dot{+} \vup$
		\item $\forall x_i \in \ar.\resets: \:\: \zconf'.\avg[i] = I_0$
	\end{itemize}
\end{itemize}
\end{definition}

\smartpar{Paths.} A sequence $\zconf_0, \ar_0, \zconf_1,\ldots,\zconf_{k-1},\ar_{k-1}, \zconf_k$ of alternating $(0,1)$-configur-ations and abstract rules is a \emph{path} of $\zcsa$ if $\forall i < k$ we have $(\zconf_{i},\ar_i,\zconf_{i+1}) \in \mTz$. We denote by $Paths(\zcsa)$ the set of all paths of $\zcsa$.

We say that a $01$-configuration $\zconf$ satisfies a reachability specification $L_{spec}$, denoted $\zconf \models L_{spec}$, if
for all $i \in L_{=0}$, $\zconf.\vk^z[i]=0$, and for all $i \in L_{>0}$, $\zconf.\vk^z[i]>0$.
We say that $\zcsa$ satisfies $L_{spec}$, denoted $\zcsa \models L_{spec}$, if there is a non-spurious path of $\zcsa$ that ends in $\zconf$ with $\zconf \models L_{spec}$.

Together with Lemma~\ref{lem:reach-cov-pathequiv}, the following lemma shows that our $(0,1)$-abstraction is precise for reachability and coverability specifications.
In other words, we have $\zcsa \models L_{spec}$ if and only if $\csa \models L_{spec}$, given that $L_{spec}$ is a reachability or coverability specification.

\begin{restatable}{lemma}{lemZcsacs}\label{lem:zcsacs}
Let $\absA=(L, \mI, \absShared, \Pi,\absRules)$ be an abstract TA, and $L_{spec}$ a reachability or coverability specification. We assume w.l.o.g. that $\mI = \{l_0\}$.%
	
	Then, there exists a path $\zconf_0, \ar_0, \zconf_1, \ldots, \ar_{n-1}, \zconf_n \in Paths(\zcsa)$ such that $\zconf_n \models L_{spec}$ if and only if there exists a path $\absconf_0, \ar_0, \absconf_1, \ldots, \ar_{m-1}, \absconf_m \in Paths(\acsa)$ such that $\absconf_m \models L_{spec}$.
\end{restatable}

%% file: reachability.tex
\subsection{Parameterized Reachability Algorithm (PRA)}
\label{sec:checking-reachability}

Our algorithm PRA for solving the parameterized reachability problem accepts two input parameters: a $01$-counter system $\zcsa$ and a finite set of error configurations $ERR \subseteq \zConfs$.
It outputs either ``The system is safe'' to indicate that no $(0,1)$-configuration in $ERR$ is reachable, or an error path of $\zcsa$.
PRA~(see App.~\ref{apx:ReachAlg}) is very similar to Algorithm \ref{alg:ParamCC}, differing primarily in the approach to computing the predecessor set. Instead of checking if two configurations are comparable, we look for equality.
Moreover, in contrast to Algorithm \ref{alg:ParamCC}, checking whether an error path $\zpi$ in $\zcsa$ is spurious is based on constraint $Reach(\zpi)$ instead of $Cover(\zpi)$, as described in Section \ref{sec:spurious}.

Just like Algorithm \ref{alg:ParamCC}, we have demonstrated soundness and completeness of PRA in a general context, along with termination, subject to the restrictions outlined in Theorem \ref{th:ParamCC-Term}~(see App.~\ref{apx:ReachAlg}).
Moreover, with minimal adjustments to the reachability algorithm, we can also check coverability within a $\zcs$~(see App.~\ref{apx:zcscover}).

%% file: experiments.tex
\section{Implementation and Experimental Evaluation}
\label{sec:experiments}

We implemented Algorithm \ref{alg:ParamCC} with explicit, unbounded integer counters and PRA symbolically, leveraging CUDD Decision Diagrams \cite{somenzi1998cudd} as BDD package. 
Both use Z3 \cite{z3} as SMT solver back-end.
We evaluated our implementations on an 
AMD Ryzen 7 5800X CPU, running at 3.8GHz with 32GiB memory of system memory. 
For comparisons with ByMC~\cite{konnov2018bymc}, the state-of-the-art model checker for threshold automata, we executed the tool in the VM provided by the authors on the same machine, with 6 out of 8 cores and 25GiB of the total system memory. To ensure an equal environment, we also executed our tool in a virtual machine with the same restrictions.\footnote{The benchmark files and a container image with our tool are available on Zenodo~\cite{artifact}.}

\begin{wraptable}[22]{r}{.5\textwidth}
\vspace{-2em}
    \caption{Comparison of execution time (in s) between ByMC \cite{konnov2018bymc} and our symbolic implementation of PRA. In the Prop. column, $A$ denotes agreement, V denotes validity, and U denotes unforgeability. \textit{TO} denotes a timeout after 1.5h.}
    \label[table]{tab:benchmark-bymc-vs-impl}
\resizebox{\linewidth}{!}{
    \begin{tabular}{  l  l r  r }
        \toprule
            \textbf{Benchmark~~~} & \textbf{Prop.} & \textbf{~ByMC} & \textbf{~~PRA} \\
        \midrule
            \textbf{aba}   & U & 0.52   & \textbf{0.12} \\
            \textbf{bcrb}  & U & 0.30   & \textbf{0.06} \\
            \textbf{bosco} & V & 42.22  & \textbf{0.73} \\
            \textbf{c1cs}  & V & \textit{\textbf{TO}} & \textbf{0.23} \\
            \textbf{cc}    & A,V & 0.24  & 0.24 \\
            \textbf{cf1s}  & V & 380.81 & \textbf{0.42} \\
            \textbf{frb}   & U & 0.19   & \textbf{0.07} \\
            \textbf{nbacg} & A & \textbf{0.18}   & 0.36 \\
            \textbf{nbacr} & V & 0.19   & \textbf{0.14} \\
            \textbf{strb}  & U & 0.17   & \textbf{0.12} \\
            \textbf{RB-bc} & A,V & \textbf{0.32} & 0.49 \\
    			\textbf{RB} & V & \textit{\textbf{TO}} & 1.03 \\
    			\textbf{RB-Simpl} & V & \textit{\textbf{TO}} & 61.77 \\
        \bottomrule
    \end{tabular}
		}
\end{wraptable}
As benchmarks in our decidable fragment, we used the following threshold-based algorithms from the literature~\cite{faulttolerantbenchmarks}: 
folklore reliable broadcast (frb)~\cite{frb}, one-step consensus with zero-degradation (cf1s)~\cite{cf1s}, consensus in one communication step (c1cs) \cite{c1cs}, consistent broadcast (strb), asynchronous byzantine agreement (aba) \cite{aba}, non-blocking atomic commit~(nbacr~\cite{nbacr} and nbacg~\cite{nbacg}), condition-based consensus (cc)~\cite{cc}, and byzantine one step consensus (bosco)~\cite{bosco}.
Moreover, we considered the following parts of the Red Belly blockchain~\cite{redbellydbft}, which have already been modeled as TA in \cite{BGKLTW22}:
the broadcast protocol~(RB-bc), the one-round consensus protocol~(RB), and a simplified one-round consensus protocol~(RB-Simpl).

We verified interesting safety specifications: \emph{agreement}~(consistent decisions among correct processes), \emph{validity}~(the value that has been decided must have been proposed by some process ), and \emph{unforgeability}~(if all correct processes have an initial value of 0, then no correct process ever accepts).

\cref{tab:benchmark-bymc-vs-impl} compares execution times of ByMC and our implementation of PRA. It shows that PRA significantly outperforms ByMC on all benchmarks except two.
Except for Red Belly protocols, our non-symbolic implementation can verify all the benchmarks above in less than 2 minutes.

For the undecidable fragment, we used our extended model of TA to model the following FTDAs from the literature: reliable broadcast~\cite{relbrd}, k-set agreement~\cite{floodmin}, multi-round simplified Red Belly blockchain consensus~\cite{redbellydbft}, multi-round full Red Belly blockchain consensus~\cite{redbellydbft}, and phase king consensus~\cite{king}.
Note that, in \cite{BGKLTW22}, only the one-round protocols were modeled, and only the simple version was verified.
In contrast, we can model the multi-round versions, and were able to verify all (single- and multi-round) versions except one.

\begin{wraptable}[14]{r}{.34\textwidth}
\vspace{-1.5em}
    \caption{Execution time (in s) for PRA on multi-round protocols. \textit{TO} denotes a timeout after 2h.}
    \label[table]{tab:benchmark-our-impl}
\resizebox{\linewidth}{!}{
    \begin{tabular}{  l  r  }
        \toprule
            \textbf{Benchmark~~~} & \textbf{PRA} \\
        \midrule
           
		    \textbf{multiR-floodMin} & 0.3 \\
		    \textbf{multiR-RelBrd} & 0.1 \\
		    \textbf{multiR-RB-Simpl} & 6859\\
	    		\textbf{multiR-RB} & \textit{TO} \\
		    \textbf{phase-king-buggy}  & 120\\
    			\textbf{phase-king-partial} & 300\\
        \bottomrule
    \end{tabular}
		}
\end{wraptable}
For the aforementioned TAs, we conducted our experiments on a machine
with  2x AMD EPYC 7773x - 128 Cores, 256 Threads and 2TB RAM.
\cref{tab:benchmark-our-impl} presents execution for running our symbolic implementation on all the aforementioned multi-round protocols. For the phase-king protocol, PRA was able to locate a bug in an incorrect model of the algorithm within $2$ minutes, and was also able to prove partial correctness properties, for example that consensus is actually reachable, in around 5 minutes. 
We verified \emph{validity} for all the remaining benchmarks, and for floodmin we verified additionally \emph{agreement}.
During benchmarking, memory usage peaked at 73 GB.

%% file: related-work.tex
\section{Related work} 
After presenting an approach for verifying FTDAs based on abstraction~\cite{john2013parameterized}, 
Konnov et al. \cite{konnov2017short,k17completeness,konnov20172,konnov2018bymc}
developed several approaches and algorithms specifically tailored for verifying the safety and liveness properties of threshold automata.
Starting from an acyclic CFA~(control flow automaton), they construct in \cite{john2013parameterized} a finite counter system using parametric interval abstraction for variables and counters.
To refine the abstraction, they detect spurious transitions using ordinary model checking and user-defined invariants.
Subsequent works~\cite{k17completeness, konnov20172} improve on the efficiency when verifying reachability properties in TAs.
In~\cite{k17completeness}, they showed the existence of an upper bound on the distance between states within counter systems of TAs, hence, demonstrating the completeness of bounded model checking.
In~\cite{konnov20172}, the authors use partial order reduction to generate a finite set of sequences comprising sets of guards and rules. Each of these sequences represents a possibly infinite set of error traces. An SMT solver is employed to validate the existence of concrete error traces.
Extending
~\cite{konnov20172}, Konnov et al. presented an approach capable of detecting lasso-shaped traces that violate a given liveness property~\cite{konnov2017short}.
The latter two approaches have been implemented in the tool ByMC \cite{konnov2018bymc}.
To enable the functionality of all aforementioned approaches, the authors found it necessary to impose constraints on threshold automata. This involved explicitly prohibiting cycles and variable decrements.

A verification tool for parameterized distributed algorithms has been introduced in~\cite{thomas2023pylta}. The tool relies on layered threshold automata~\cite{bertrand2021guard} as a system model, which can be seen as an infinitely repeating threshold automata. However, the tool requires users to identify layers~(rounds) in the model, their sequence~(infinite interleaving or lasso-shaped sequences), and to provide  predicates.

Decidability and the complexity of verification and synthesis of threshold automata have also been studied in~\cite{esparza2020TAcomplexity}.
Their decision procedure is based on an SMT encoding of potential error paths, where in general the size of the SMT formula grows exponentially with the length of the paths.
While having achieved good results with some heuristics that avoid this exponential blow-up in practice, we believe that these heuristics would not work for threshold automata with decrements and/or resets\footnote{We could not verify this conjecture since their code is not publicly available.}.
Moreover, their method requires a bound on the number of changes in the valuation of thresholds. However, such a constraint does not apply in the presence of decrements and/or resets.

%% file: conclusion.tex
\section{Conclusion}
\label{sec:conclusion}
In this paper, we have introduced an extension of the computational model known as threshold automata, to support decrements and resets of shared variables.
This extension in general comes at the cost of decidability, even for simple state-reachability properties.
We developed a semi-decision procedure for this extended notion of TA, supporting not only the simple coverability properties from the TA literature, but also general coverability properties as known from Petri nets or well-structured systems.
To support also reachability properties, we presented an additional abstraction, called $(0,1)$-abstraction, which is the basis for a semi-decision procedure for reachability properties of extended TAs.

We have implemented our techniques and evaluated them on examples from the literature, and on several round-based algorithms that cannot be modeled with canonical TAs.
We show that our semi-decision procedure can find bugs in a faulty protocol and prove correctness of protocols, even outside the known decidable fragment.
Moreover, on a set of benchmarks in the decidable fragment, it matches or outperforms the TA model checker ByMC~\cite{konnov2018bymc}.
\begin{credits}
\subsubsection{\ackname} T. Baumeister and P. Eichler carried out this work as members of the Saarbr\"ucken Graduate School of Computer Science.
This research was funded in whole or in part by the German Research Foundation (DFG) grant 513487900 and the Luxembourg National Research Fund (FNR) grant C22/IS/17432184. For the purpose of open access, and in fulfilment of the obligations arising from the grant agreement, the author has applied a Creative Commons Attribution 4.0 International (CC BY 4.0) license to any Author Accepted Manuscript version arising from this submission.
\end{credits}

%% file: proofs.tex
\section{Proofs}\label{apx:proofs}

\lemGenCovPathequiv*
\begin{proof}
$\Rightarrow$:
Let $\bpi$ be a non-spurious path of $\acsa$ that covers $\conf_{spec} \in \Confs_{spec}$. Then $Concretize(\bpi) \land Cover(\bpi)$ is satisfiable and its satisfying assignment is a path in $\csa$ that covers $\conf_{spec}$.
\begin{itemize}
	
	\item $
	RC \land \textstyle\sum_{j \in \mI } \conf_0.\vk[j] = N(n,t,f)  $, makes sure that we have admissible parameters and we start from a valid initial configuration.
	\item $\forall i < m \; \conf_{i}.\vk[r_i.\from] = c_i + \conf_{i+1}.\vk[r_i.\from] \land
	\conf_{i+1}.\vk[r_i.\to] = c_i + \conf_{i}.\vk[r_i.\to] $, ensures that, for each local transition, exactly one process makes a move.
	\item $ \forall i < m \; \forall x_j \in \Shared \;
	x_j \notin r_i.\tau \implies
	\conf_{i+1}.\vg[j] = \conf_{i}.\vg[j] + c_i \cdot r_i.\vup[j]$, guarantees correct updates for all the variables that haven't been reset.
	\item $\forall i < m \; \forall x_j \in r_i.\resets \; \conf_{i+1}.\vg[j] = 0$, resets all the variables in $r_i.\resets$.
	\item $\forall i < m \;
	\forall x_j \in \Shared \;  \conf_{i}.\vg[j] \in \absconf_{i}.\avg[j] \land \conf_{i+1}.\vg[j] \in \absconf_{i+1}.\avg[j]
	$
	\item[] $\land$ $\forall i < m \; c_i > 1 \implies ((\conf_{i+1}.\vg - r_i.\vup) \models r_i.\varphi)$, ensures that a guard is satisfied before it is executed.
\end{itemize}

$\Leftarrow$:
Let $\pi$ be a path of $\csa$ that covers $\conf_{spec} \in \Confs_{spec}$.
Let us construct 
$\bpi=\absconf_0,\ar_0,\ldots,\ar_{m-1},\absconf_m$
from $\pi$ by replacing concrete rules and configurations with their corresponding abstract versions, and by then replacing any sub-sequence of
repeating abstract configurations and rules
$\absconf, \ar, \absconf, \ar, \absconf$
with $\absconf, \ar, \absconf$.
Then $\absconf_0 \in \absConfs_0$, $\forall i \: \conf_{spec}.\vk[i] \leq \absconf_m.\vk[i]$, and $\bpi$ is a path of $\acsa$.

\qed
\end{proof}

\lemAbcsWSTS*
\begin{proof}
The order $\ord$ is reflexive and transitive, and it has no infinite antichain due to the fact that $\leq$ is a wqo.
To see that $\ord$ is strongly compatible with $\Trans$,
let $\transition{(\vk_1,\avg_1)}{\ar}{(\vk'_1,\avg'_1)}$, and let  $(\vk_1,\avg_1)\ord(\vk_2,\avg_2)$.
Then we have $\transition{(\vk_2,\avg_2)}{\ar}{(\vk'_2,\avg'_2)}$ due to the fact that $\avg_1 = \avg_2 \land \vk_1 \leq \vk_2$.
We also have $(\vk'_1,\avg'_1)\ord(\vk'_2,\avg'_2)$ by resolving the non-determinism for obtaining $\avg_2'$ in the same way as in the transition $\transition{(\vk_1,\avg_1)}{\ar}{(\vk'_1,\avg'_1)}$~(i.e., whether to change an interval or not).
\qed
\end{proof}

\lemPredbasis*
\begin{proof}
	Let $E' \subseteq \absConfs$ be finite.
Note that $pred(\nuparrow E')$ is upward-closed due to the strong compatibility of $\ord$ established in the proof of Lemma~\ref{lem:abcsWSTS}, and therefore it has a finite minimal basis, denoted $minBasis(pred(\nuparrow E'))$.
 
Let $\vu_i$ be the unit vector with $\vu_i(i) = 1$ and $\vu_i(j) = 0$ for $j \neq i$. Consider the set $CPredBasis(E')$, defined as
\[ 
\left\{
\begin{array}{l|l}
	(\vk,\avg) \in \absConfs~ 												& ~\exists (\vk',\avg') \in E', \exists \ar = (l_i,l_j,\bphi,\vup, \resets) \in \absRules \text{ such that } \\
 
													& ~~\left[ \transition{(\vk,\avg)}{\ar}{(\vk',\avg')} \lor   \left( \transition{(\vk,\avg)}{\ar}{(\vk' + \vu_j,\avg')} \land \vk'[j] = 0 \right) \right]
\end{array} \right\}.
\]

First, note that we have $CPredBasis(E') \subseteq pred(\nuparrow E')$, which follows by definition of $CPredBasis(E')$ and the fact that, if $(\vk',\avg') \in E'$, then $(\vk' + \vu_j,\avg') \in \nuparrow E'$.
To see that $CPredBasis(E') \supseteq minBasis(pred(\nuparrow E'))$, assume towards a contradiction that it is not.
Then, there exist $(\vk_1,\avg_1) \in (minBasis(pred(\nuparrow E')) \cap \neg CPredBasis(E'))$, $(\vk'_1,\avg'_1) \in \nuparrow E'$,  and $(\vk'_2,\avg'_2) \in E'$ with $\transition{(\vk_1,\avg_1)}{\ar}{(\vk'_1,\avg'_1)}$ and $(\vk'_2,\avg'_2) \ord (\vk'_1,\avg'_1)$.
We differentiate between two cases:
\begin{itemize}
\item Case $\vk'_2[j] = 0$: Let $\vk_2 = \vk'_2 + \vu_i$, and $\avg_2 = \avg_1$.
Then $\avg_2 \models \ar.\bphi$, and $((\vk_2,\avg_2),\ar,(\vk'_2+\vu_j,\avg'_2)) \in \absRules$.
Hence $(\vk_2,\avg_2) \in CPredBasis(E')$, and $(\vk_2,\avg_2) \ord (\vk_1,\avg_1)$.
Contradiction.
\item Case $\vk'_2[j] > 0$: Let $\vk_2 = \vk'_2 + \vu_i - \vu_j$, and $\avg_2 = \avg_1$.
Then $\avg_2 \models \ar.\bphi$, and $((\vk_2,\avg_2),\ar,(\vk'_2,\avg'_2)) \in \absRules$.
Hence $(\vk_2,\avg_2) \in CPredBasis(E')$, and $(\vk_2,\avg_2) \ord (\vk_1,\avg_1)$.
Contradiction.	\qed
\end{itemize}
\end{proof}

\thParamCCTerm*
\smartpar{Proof sketch.}
Suppose we have $\forall \ar \in \absRules: \ar.\vup \in |\Nat_0^{|\Shared|} \land \ar.\resets = \emptyset$.
In this case, once a rule changes its status~(enabled/disabled), it will retain it forever.
Suppose $\bpi^e = \absconf^e_0, \ar_0,\ldots,\ar_{m-1}, \absconf^e_m$ is a path in $errGraph$ and let  $C_e = \absconf^e_i,\ar_i,\ldots,\ar_j, \absconf^e_j$ be a sub-sequence of $\bpi^e$ with $\absconf^e_i = \absconf^e_j$.
Note that a cycle $\absconf^e_i,\ar_i,\ldots,\ar_j, \absconf^e_i$ in $errGraph$ does not necessarily correspond to a concrete cycle.
It just indicates that the sequence of transitions $r_i,\ldots,r_j$ may be executed an arbitrary number of times. 
In the general case, the algorithm checks if $\bpi^e$ is spurious~(see Section \ref{sec:spurious}) and whether paths, formed by unfolding $C_e$ in $\bpi^e$ an arbitrary number of times, are spurious.
However, in the absence of shared variable decrements and resets, it suffices to check $\bpi^e$. This is because, regardless of how we unfold the cycle $C_e$, the initial configuration, the final configuration, and all intermediate configurations will consistently have the same vector $\avg$.
Please note that if the TA contains cycles with upper guards, then while only a finite number of processes may traverse it, they can do so an unbounded number of times. Therefore, the strategy outlined above may not suffice in such cases.
Consequently, in such a case, we need to use a simplified variant of the reachability formula used in \cite{esparza2020TAcomplexity}.
This formual would be relatively concise, since it uses only the transitions that appear on the cycle of the error path and leverages the total order on the transitions of the TA cycle.

\qed
It's worth noting that none of the TA benchmarks in Table \ref{tab:benchmark-bymc-vs-impl} have cycles.
Furthermore, all the computed error graphs for the benchmarks in \ref{tab:benchmark-our-impl} were empty.

\lemZcsacs*
\begin{proof}
$\Rightarrow$:
	
	Suppose there exists $\zpi = \zconf_0, \ar_0, \zconf_1, \ldots, \ar_{n-1}, \zconf_n \in Paths(\zcsa)$ such that $\zconf_n \models L_{spec}$. 
	Let $b$ be the number of transitions $(\zconf_k,\ar_k,\zconf_{k+1})$ in $\zpi$ with $\zconf_{k+1}.\vkz[\ar.\from] = 1$, i.e., the transitions where we keep a $1$ in location $\ar.\from$.
	 Then we construct the corresponding path $\bpi=\absconf_0,\ar_0,\ldots,\ar_{m-1},\absconf_m$ as follows:
	\begin{enumerate}
	\item Let $\absconf_0$ be the initial configuration in $\bpi$ with $\absconf_0.\vk[l_0] = 2^b$.
	\item For every transition $(\zconf_k,\ar_k,\zconf_{k+1})$ in $\zpi$:
	\begin{itemize}
	\item if $\zconf_{k+1}.\vkz[\ar_k.\from] = 0$, then we obtain $\absconf_{k+1}$ from $\absconf_k$ by executing $\ar_k$ until location $\ar_k.\from$ becomes empty. If $\zconf_{k+1}.\avg \neq \zconf_k.\avg$, then in $\bpi$ we only change the valuation of shared variables $\avg$ with the last execution of $\ar_k$.
	\item if $\zconf_{k+1}.\vkz[\ar_k.\from] = 1$, then we execute $\ar_k$ $\frac{c_k}{2}$ times, where $c_k$ is the number of processes that are in location $\ar_k.\from$
          (i.e., we move half of the processes to $\ar_k.\to$, and keep the other half in $\ar_k.\from$).
	As before, we change $\avg$ only with the last execution of $\ar_k$.
	Since we started with $2^b$ processes in the initial configuration, this is possible $b$ times, which is exactly the number of such transitions in the path.
	\end{itemize}
	\end{enumerate}
By construction of $\bpi$, the same locations in $\zconf_n$ and $\absconf_m$ are occupied, and hence, $\absconf_m \models L_{spec}$.
	
	$\Leftarrow$: 
	
	Suppose there is $\absconf_0, \ar_0, \absconf_1, \ldots \ar_{m-1} \absconf_m \in Paths(\acsa)$ such that $\absconf_m \models L_{spec}$. Then the corresponding path of $\zcsa$ is obtained by simply replacing each configuration $\absconf_j$ with a $(0,1)$-configuration $\zconf_j$ such that for each location $i$, $\absconf_j.\vk[i] \geq 1 \iff \zconf_j.\vkz[i] = 1$ and $\absconf_j.\vk[i] = 0 \iff \zconf_j.\vkz[i] = 0$. 
	The result is a valid path of $\zcsa$ with $\zconf_m \models L_{spec}$.
	\qed	
\end{proof}

\lemConstructCE*
\begin{proof}
	
	Since $\ord$ is a wqo, any infinite increasing sequence of upward-closed sets eventually stabilizes~\cite{finkel2001well}.
	Therefore, Algorithm \ref{alg:ParamCC} is guaranteed to converge to a fixed-point.
	Note that for increasing $j$, the sequence $\bigcup_{0 \leq i \leq j} E_i$, which is based on the sets $E_i$ computed by Algorithm~\ref{alg:ParamCC}, is an increasing sequence of upward-closed sets. Therefore, the algorithm will converge to a fixed-point.
	
	In order to reach a contradiction, suppose the algorithm reached a fixed-point without computing any non-spurious error path, and let $E_i,\ldots, E_0=\Confs_{spec}$ be the computed sets of predecessor bases.
	Let $\bpi=\absconf_0,\ar_0,\ldots,\ar_{j-1},\absconf_j$. Since $\bpi$ is an error path that was not computed by Algorithm \ref{alg:ParamCC}, then $\absconf_j \in {\uparrow}E_0$ and there exists $i>0$ where $\forall k > i \; 	\forall l \leq i: \: \absconf_{j-k} \not\in {\uparrow}E_l$, i.e. $\{ \absconf_0,\ldots,\absconf_{j -(i + 1)}\} \cap ({\uparrow}E_i \cup \ldots \cup {\uparrow}E_0) = \emptyset$.
	Due to procedure $\textproc{ComputePredBasis}$ we have $\forall k \leq i: \: \absconf_{j-k} \in E_l$ with $l \leq k$.
	We assume that $i < j$ since other cases are trivial.
	Executing the procedure $\textproc{ComputePredBasis}(E_i, visitedTrans)$~(based on Lemma \ref{lem:predbasis}) must have resulted in an empty set.
	Then, $\absconf_{j-(i+1)} \in E_l$ with $l \leq i$.%
	Furthermore, the algorithm guarantees also that the transition $(\absconf_{j-(i+1)},\ar_{j-(i+1)},\absconf_{j-i}) \in visitedTrans$ is computed~(check Line \ref{line:appendAll}).
	Therefore, due to Corollary \ref{cor:compStatesPred}, $\bpi$ is a path in $errGraph$.
	Contradiction.
	\qed
\end{proof}

%% file: ReachAlg.tex
\section{Parameterized Reachability Algorithm}\label{apx:ReachAlg}

This section introduces our algorithm for solving the parameterized reachability problem (see Definition \ref{def:reachability}).
Algorithm \ref{alg:ParamMC} accepts two input parameters: a $01$-counter system $\zcsa$ and a finite set of error configurations $ERR \subseteq \zConfs$.
It outputs either ``The system is safe'' to indicate that no $(0,1)$-configuration in $ERR$ is reachable, or an error path of $\zcsa$.

Algorithm \ref{alg:ParamMC} is very similar to Algorithm \ref{alg:ParamCC}, differing primarily in the method for computing the predecessor set.
As before, error paths may be spurious, and the algorithm needs to keep track of all possible paths on which a configuration may be reached.
In contrast to Algorithm \ref{alg:ParamCC}, checking whether error path $\zpi$ in $\zcsa$ is spurious is based on constraint $Reach(\zpi)$ instead of $Cover(\zpi)$, as described in Section \ref{sec:spurious}.

\input{rcAlgorithm}

Combining Lemmas \ref{lem:ConstructCE} and \ref{lem:zcsacs} yields the following corollary:

\begin{corollary}[Soundness, Completeness]
	Algorithm \ref{alg:ParamMC} is sound and complete, i.e., Algorithm \ref{alg:ParamMC} computes a non-spurious error path if and only if $\csa$ includes a reachable configuration that violates a given reachability specification.
\end{corollary}

\smartpar{Termination}. Unfortunately, Algorithm \ref{alg:ParamMC} may compute an error graph with cycles ($errGraph$), and checking for non-spurious error paths may not terminate because there may be infinitely many paths to check.

\begin{theorem}\label{th:ParamMC-Term}
	If $\forall \ar \in \absRules: \ar.\vup \in |\Nat_0|^{|\Shared|} \land \ar.\resets = \emptyset$, then Algorithm \ref{alg:ParamMC} terminates.
\end{theorem}

The proof is very similar to Theorem \ref{th:ParamCC-Term}, and thus has been omitted. Similarly, in this scenario, the algorithm terminates if $\csa \models L_{spec}$ or if $errGraph$ is acyclic, empty, or all its cycles have neither decrements nor resets.

%% file: rcAlgorithm.tex
\begin{algorithm}[h]
	\caption{Parameterized Reachability Checking (PRA)}\label{alg:ParamMC}
	
	\begin{algorithmic}[1]		
		\Procedure{CheckReachability}{\emph{$01$-Counter System} $\zcs$, $ERR$}\par
		 \hskip\algorithmicindent{\color{brown} //same approach as $\textproc{CheckCoverability}$ in Algorithm \ref{alg:ParamCC} but
		instead of\par
		 \hskip\algorithmicindent //$\textproc{ComputePredBasis}$ we have $\textproc{ComputePred}$}
		\EndProcedure
		\Procedure{ComputePred}{$E_{i-1}$,$visitedTrans$}
		\State $TrPred(E_{i-1}) \gets \{(\zconf,\ar,\zconf') \in \mT^z \mid \zconf' \in E_{i-1}\}$\label{line-rc:TrPred}
		\State $E_i= \{\zconf \mid \exists (\zconf,\ar,\zconf') \in TrPred(E_{i-1})\}\label{line-rc:Ei}$%
		\State $visitedTrans.appendAll(TrPred(E_{i-1}))$ \label{line-rc:appendAll}
		\State $finalE_i \gets \emptyset$
		\ForAll{$\zconf \in E_i$} {\color{brown} //remove visited equal configs to ensure termination}
		\If{$\exists j \leq i, \zconf_e \in E_j$ s.t. $\zconf_e = \zconf$}
		\State $visitedTrans(\zconf_e).append(\{(\zconf_e,\ar,\zconf') \mid (\zconf,\ar,\zconf') \in  TrPred(E_{i-1}) \})$\label{line-rc:updateMap1}		
		\Else
		\State $finalE_i.append(\zconf)$
		\EndIf
		\EndFor
		\State \textbf{return} $finalE_i,visitedTrans$  
		\EndProcedure
	\end{algorithmic}	
\end{algorithm}

Procedure $\textproc{ComputePred}$ takes as argument the set $E_{i-1}$, and the set of visited transitions $visitedTrans$.
We begin by obtaining the set of predecessors, denoted as $E_i$, from the incoming transitions of the configurations in $E_{i-1}$~(Lines \ref{line-rc:TrPred}-\ref{line-rc:Ei}).
The set $E_i$ may contain configurations that are equal to configurations that have been already explored.
Hence, after adding all computed transitions to $visitedTrans$~(Line \ref{line-rc:appendAll}), we check if there is a configuration $\zconf_e$ with $\zconf_e = \zconf$.
If such a configuration is found, we add $\zconf$'s transitions to $\zconf_e$ (Line \ref{line-rc:updateMap1}).

%% file: zcsCover.tex
\section{Checking Coverability in $\zcsa$}\label{apx:zcscover}

Let $\zord \; \subseteq \zConfs \times \zConfs$ be the binary relation defined by:
$$(\vkz,\avg) \zord (\vkz{'},\avg') \iff \vkz \leq \vkz{'} \land \avg = \avg'$$
where $\leq$ is the component-wise ordering of vectors.
Then, we say that a $(0,1)$-configuration $\zconf$ satisfies a coverability specification $L_{spec}$~(see Definition \ref{def:coverability}), denoted $\zconf \models L_{spec}$, if  for all $i \in L_{>0}$, $\zconf.\vkz[i]>0$.
Thus, out of a coverability specification $L'_{spec} = L_{>0}$, we can generate a reachability specification $L_{spec} = L_{>0} \cup L_{=0}$ and use Algorithm \ref{alg:ParamMC} to check it.
Alternatively, a variant of Algorithm \ref{alg:ParamCC} can be used, in which the $\textproc{ComputePredBasis}$ procedure uses $\zord$ in place of the previously used $\ord$.